\documentclass[journal,transmag]{IEEEtran}
\ifCLASSINFOpdf
\else
\fi
\usepackage{amsmath, amsthm, amscd, amsfonts, epsfig, amssymb, longtable, soul, booktabs, cuted}
\usepackage[utf8]{inputenc}

 \newtheorem{theorem}{Theorem}
\newtheorem{lemma}{Lemma}

\newtheorem{proposition}{Proposition}
\newtheorem{definition}{Definition}

\newcommand{\rdt}{$R_{(d,T, b)}$}
\newcommand{\rdts}{$R_{(d,T, b)}$ }
\newcommand{\dtbs}{$(d,T,b)$-$MP$ }
\newcommand{\dtb}{$(d,T,b)$-$MP$}

\begin{document}
%
\title{Assessing the performance of smart grid communication networks under both time and budget constraints}


\author{\IEEEauthorblockN{Majid Forghani-elahabad\IEEEauthorrefmark{},
\IEEEauthorblockA{\IEEEauthorrefmark{}Center of Mathematics, Computing, and Cognition - Federal University of ABC, Santo André, SP, Brazil}}
\thanks{Manuscript received ...; revised .... 
Corresponding author: M.~Forghani-elahabad (email:~m.forghani@ufabc.edu.br, phone: +55(11)49968330).}}

%



\IEEEtitleabstractindextext{%
\begin{abstract}
The smart grid concept has emerged to address the existing problems in the traditional electric grid, which has been functioning for more than a hundred years. The most crucial difference between traditional grids and smart grids is the communication infrastructure applied to the latter. However, coupling between these networks can increase the risk of significant failures. Hence, assessing the performance of the smart grid communication networks is of great importance and thus is considered here. As transmission time and cost play essential roles in many real-world communication networks, both time and budget constraints are considered in this work. To evaluate the performance of communication networks, we assume that the data is transmitted from a source to a destination through a single path. We propose an algorithm that computes the exact probability of transmitting $d$ units of data from the source to the destination within $T$ units of time and the budget of $b$. The algorithm is illustrated through a benchmark network example. The complexity results are also provided. A rather large-size benchmark, that is, Pan European topology, along with one thousand randomly generated test problems are used to generate the experimental results which show clearly the superiority of our proposed algorithms to some existing algorithm in the literature.
\end{abstract}

}

\maketitle

\IEEEdisplaynontitleabstractindextext

%
\IEEEpeerreviewmaketitle

\subsection*{Acronyms}
\begin{tabular}{ll}	
	SGCN&Smart Grid Communication Network\\
	MFN&Multistate Flow Network\\
	MP&Minimal path\\
	SSV&System State Vector
\end{tabular}
\subsection*{Nomenclature}
\begin{itemize}	
	\item For two SSVs $ X = (x_1, x_2, \cdots, x_m) $ and $ Y = (y_1, y_2, \cdots, y_m) $, we say that $ X\leq Y $ if $ x_i\leq y_i $, for any $ i = 1, 2, \cdots, m $. And $ X<Y $ when $ X\leq Y $ and there exists at least one $ 1\leq j\leq m $ such that $ x_j<y_j $. For instance, assuming $ X=(2, 3, 1)$, $Y = (1, 3, 1)  $ and $ Z=(2, 2, 2) $, we have $ Y< X $, $ Y\nless Z $ and $ Z\nless Y $.
	\item An SSV, say $ X\in \Psi $, is said to be a minimal vector in $ \Psi $ if there is no any $ Y\in \Psi $ such that $ Y<X $. 	For instance, all the vectors in $ \{$(1, 2, 3), (2, 3, 1), (3, 2, 1)$\} $ are the minimal vectors.
\end{itemize}

\subsection*{Assumptions}
\begin{enumerate}
	\item The capacity of arc $ a_i\in A $ takes random integer values from $ \{0, 1, \cdots, M_i\} $ according to a given probability distribution function, for $ i = 1, 2, \cdots, m $.
	\item The arcs' capacities are statistically independent one from the other arc.
	\item Every node is perfectly reliable, i.e., deterministic.
	\item Flow in the network satisfies the flow conservation law.
	\item The data is transmitted from a source node to a destination node through a single path.
	\item All the minimal paths of the network are given in advance.
\end{enumerate}

\section{Introduction}
%
%
%
%

%

\IEEEPARstart{A}{} communication network in a smart grid is a network of networks that may utilize several different
communication technologies supported by two main communications media, that is, wired and 
wireless, to be expanded deep into the distribution grids~\cite{bakken2014smart}, \cite{wen2014form}. Thus, the communication network will be responsible for gathering and routing data, monitoring every node, and acting upon the received data. On the one hand, wireless communications have some advantages over wired technologies, such as low-cost infrastructure and ease of connection to difficult or inaccessible areas. However, the nature of the transmission path may cause the signal to attenuate~\cite{gungor2010opportunities}. 
On the other hand, wired solutions generally do not have interference problems, and their functions are not dependent on batteries, as wireless solutions often do~\cite{wen2014form}. 

A communication network is usually modeled in one of the two following ways:
(1) The network components have a binary behavior, i.e., they can only have a fully
working (operational) or wholly failed state. In this case, the network also has two states, and they
are called binary-state networks.
(2) The components, and consequently the network, can have more than two states. In this case,
the network is called a stochastic-flow network, or a multistate flow network (MFN)~\cite{forghani2019iise}. Generally, to model a smart grid communication network as an MFN, all the components in the network that should communicate with the other ones such as smart meters, appliances, or control centers can be considered as a node, and each communication link between two nodes can be considered as an arc. 

In a communication network, we may have one or more of the three kinds of two-way
communication: (1) \textit{two}-terminal, (2) $ k $-terminal, and (3) \textit{all}-terminal communication. 
The \textit{two}-terminal 
communication is a case in which two specific components communicate with each
other, like the communication between each smart meter and the corresponding center. 
The $ k $-terminal
communication is a case in which a component communicates with $ k $ other components, like the
communications between appliances at home and a smart meter. 
The \textit{all}-terminal
communication is a case where all the network components communicate to each other, like the communications between substations in a smart grid. 
The focus of this work is on the first case, that is, the \textit{two}-terminal communication.

Several indexes have been proposed in the literature to assess the performance of the MFNs such as communication networks, transportation networks, and so forth. System reliability is one of the most attractive indexes and indeed of great importance.
Many researchers have worked on the issues related to the reliability and Quality of Service (QoS) of smart grid communication networks. Wang et al.~\cite{wang2010reliability} proposed a reliability evaluation approach to analyze the reliability of Wide-Area Measurement Systems (WAMS). The method incorporates Markov modeling and state enumeration techniques and covers the backbone communication network in WAMS. Gungor et al.~\cite{gungor2010opportunities} discussed communication technologies and requirements for smart grids and presented some QoS mechanisms for smart grids. The authors also introduced some standards and provided a better understanding of the smart grid and related technologies. 

Generally, the reliability is considered as the ability to execute a defined function under specified
conditions for a known period of time~\cite{shier1991network}. Availability in a repairable communication network or reliability in a non-repairable network can be measured as the fraction of time that a specific service is available, as the fraction of data packets that successfully attain the planned destination, or generally as the probability of the existence of a demanded connection~\cite{hines2014smart}. 
An important issue in the smart grid communication network, from the point of view of quality management, is lessening the transmission time or sending data within a given time. Therefore, many researchers have considered the time constraints in assessment of such networks~\cite{forghani2015ress}, \cite{lin2010reliability}, \cite{forghani2015ieee}, \cite{yeh2015fast}. Furthermore, transmission cost and limited budget are other key issues that play important roles in the real-world systems. 
Therefore, the \textit{two}-terminal reliability index in this work, denoted by \rdt, is defined as the probability of transmitting $ d $ units of data from a source to a destination through a single path in the network within $ T  $ units of time and budget $ b $.
Many algorithms have been proposed in the literature to compute the network reliability with or without the time and budget constraints which are usually based on the minimal cuts~
\cite{forghani2013ijor}, \cite{forghani2014ieee}, \cite{forghani2016amm}, \cite{forghani2019ijrqse}, \cite{forghani2019jcs}, \cite{niu2017evaluating-MC}, \cite{niu2019new-MC}, \cite{mansourzadeh2014comparative}, \cite{yeh2008fast, yeh2015new, huang2019reliability, bai2018reliability, datta2019evaluation} 
or minimal paths~
\cite{forghani2019iise}, \cite{forghani2015ress}, \cite{lin2010reliability}, \cite{forghani2015ieee}, \cite{chen2017search}, \cite{forghani2017ieee}, \cite{forghani2019assa}, \cite{jane2017distribution}, \cite{lin1995reliability}, \cite{lin2003extend}, \cite{niu2020capacity}, \cite{niu2020finding-MP}, \cite{yeh2018fast-MP}, \cite{yeh2019new-MP}, \cite{yeh2021novel}, \cite{forghani2020ijor}. Although many studies have been made so far in this regard, the problem of assessing the reliability of communication networks is a well-known NP-hard combinatorial problem~\cite{levitin2007computational}, and hence the research continues.
This work proposes an algorithm to compute the exact amount of \rdts for a smart grid communication network. We illustrate the algorithm through an example and compute the complexity results. Then, several numerical results are conducted through known large-size benchmark and one thousand randomly generated test problems to show the practical superiority of our proposed algorithm in comparison with some available algorithm in the literature. 
The rest of this work is organized as follows. Section~\ref{sec-preliminaries} states the required notations and basic results on the problem. The proposed algorithm is given in Section~\ref{sec-algorithm}. We provide the illustrative example and complexity results in Section~\ref{sec-example-complexity}. Section~\ref{sec-experimental results} conducts several numerical results to show the practical efficiency of the proposed algorithm. Finally, we conclude the work in Section~\ref{sec-conclusions}.

\section{Network model and basic results}\label{sec-preliminaries}
\subsection{Preliminaries}
An architecture of a smart grid communication network is depicted in Fig.~\ref{fig1:smartgrid-architecture} taken from~\cite{powergrid}. One observes that the corresponding communication network to the smart grid given in Fig.~\ref{fig1:smartgrid-architecture} is the one shown in Fig.~\ref{fig2:communication-network}. This network has $ 10 $ nodes and $ 17 $ arcs which can be considered a multistate flow network (MFN). It shows how one can study a smart grid communication network (SGCN) as an MFN. From now on, we say network instead of SGCN or MFN.

\begin{figure*}
	\centering
	\includegraphics[width=.75\linewidth]{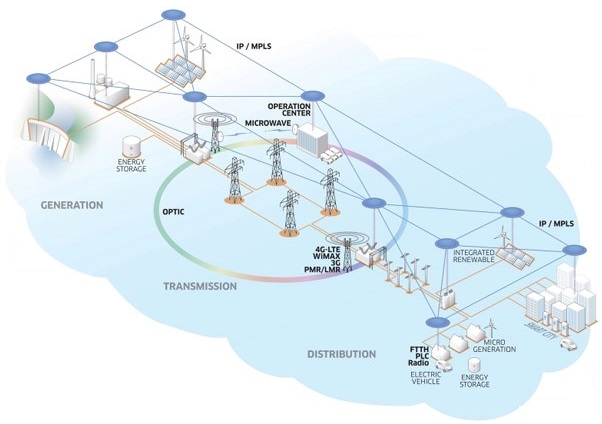}
	\caption{An architecture of a smart grid communication network taken from~\cite{powergrid}.} \label{fig1:smartgrid-architecture}
\end{figure*}

\begin{figure}
	\centering
	\includegraphics[width=.75\linewidth]{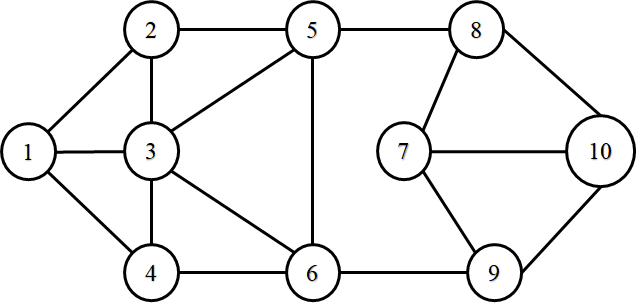}
	\caption{The corresponding communication network to the smart grid of Fig.~\ref{fig1:smartgrid-architecture}.} 
	\label{fig2:communication-network}
\end{figure}
Let $ G(N, A, M, L, C) $ be a network with the set of nodes $ N= \{1, 2, \cdots , n\} $, the set of arcs $ A = \{a_1, a_2, \cdots, a_m\} $, the maximum capacity vector $ M = (M_1, M_2, \cdots, M_m) $  in which $M_i$ is an integer-valued number and denotes the maximum capacity of arc $a_i$, for $i=1, 2, \cdots, m$, the lead time vector $L=(l_1,l_2,\cdots,l_m)$ in which $l_i$ is an integer-valued number and denotes the lead time of arc $a_i$, for $i=1, 2, \cdots, m$, and the cost vector $ C = (c_1, c_2, \cdots, c_m) $ with $ c_i $ denoting the transmission cost of $ a_i $ for sending each unit of data, for $ i = 1, \cdots, m $. Hence, $ n$ and $m $ are respectively the number of nodes and arcs in the network. Let also the nodes $ 1 $ and $ n $ be  respectively the source and destination (sink) nodes. For example, Fig.~\ref{fig3:benchmark} depicts a network with $N = \{1, 2, 3, 4, 5\}$ and $ A  = \{a_1, a_2, \cdots, a_8\} $. One may consider $M=(5, 3, 4, 3, 2, 4, 5, 3)$, $L=(2, 2, 3, 1, 2, 2, 3, 1)$ and $ C = (1, 2, 2, 1, 3, 2, 3, 4) $ as the maximum capacity, lead time and transmission cost vectors in this network. Hence, for instance, at any time no more than $ 5 $ units of data can be sent simultaneously through arc $ a_1 $ as $ M_1=5 $. Moreover, as $ l_1=2 $, it takes two units of time to transmit data from the initial node of $ a_1 $, that is, $ 1 $, to its terminal node, that is, $ 2 $. Also, transmitting any units of data on $ a_1 $ will cost $ 1 $ unit of currency as $ c_1=1 $. 
Although the lead time and transmission cost of any arc is fixed, the capacities of the arcs are not fixed and may vary due to the operations activities, failures, maintenance needs, and so forth. Hence, the current capacity of arc $a_i$ is denoted with $x_i$ which takes values from $\{0,1,\cdots,M_i\}$ according to a given probability distribution function, for $i=1, 2, \cdots, m$, and thus $X=(x_1, x_2, \cdots, x_m)$ denotes the current capacity vector of the network. For instance, $X^\star = (3, 3, 4, 2, 2, 4, 3, 3)$ can be a current SSV  for Fig.~\ref{fig3:benchmark} under which for example a maximum of $ 4 $ units of data can be sent through arc $ a_3 $ per unit of time. Moreover, it takes $ l_3=3 $ units of time to transmit these $ 4 $ units of data from node $ 1 $ to node $ 3 $ through this arc.
Now, if we want transmit $10$ units of data from nodes~$1$ to $3$ on $ a_3 $ under $ X^\star $, it needs a budget of $b= 10\times 2 = 20 $ and will take $3+\lceil 10/4\rceil = 6$ units of time, where $ \lceil \alpha \rceil $ is the smallest integer number not less than $ \alpha $. Starting the transmission from node~$1$,  $x_3=4$ units of data can be pumped into $ a_3 $ per unit of time and would take three units of time for the first data arrives at the terminal node. Afterward, $x_3=4$ units of data will be pumped out from the terminal node each unit of time, and consequently it will take $\lceil 10/4\rceil = 3$ units of time since then until all of the required data is transmitted to the terminal node. Therefore, it takes totally $3+\lceil 10/4\rceil = 6$ units of time to complete the transmission. As a result, the transmission time and cost for sending $ d $ units of data on $ a_i $ with the current capacity of $ x_i $, lead time of $ l_i $ and transmission cost of $ c_i $ are respectively equal to
\begin{align}\label{arc-lead-time}
	t= l_i+\lceil \frac{d}{x_i}\rceil &&\text{and} &&c = d\times c_i.
\end{align}

\begin{figure}
	\centering
	\includegraphics[width=0.5\linewidth]{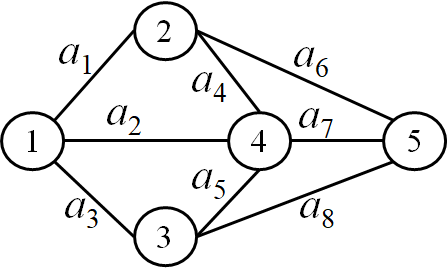}
	\caption{A benchmark network example taken from~\cite{forghani2019ijrqse}.} 
	\label{fig3:benchmark}
\end{figure}
 Pan-European-Network
A path is a set of adjacent arcs through which the data can be transmitted from the source node~to the destination node, and a minimal path (MP) is a path whose proper subsets are not path anymore. For instance, in Fig.~\ref{fig3:benchmark}, $ P =\{a_1, a_4, a_2, a_3, a_8\} $  is a path from node $ 1 $ to node $ 5 $ but not an MP because its proper subset, that is, $ P' = \{a_3, a_8\} $, is still a path. However, $ P' = \{a_3, a_8\} $ is an MP. Without loss of generality, let $ P_j=\{a_1, a_2, \cdots, a_{nj}\} $ be a path in the network with current SSV of $ X = (x_1, x_2, \cdots, x_m) $. The capacity of $ P_j $ under $ X $ is equal to 
\begin{equation*}
KP_j(X)~=~\min\{x_{1}, x_{2}, \cdots, x_{nj}\},
\end{equation*} 
the lead time of this path is equal to 
\begin{equation*}
	LP_j = \sum_{r=1}^{nj} l_{r},
\end{equation*}
and finally the transmission cost of this path per unit of data is equal to 
\begin{equation*}
	CP_j = \sum_{r=1}^{nj} c_{r}.
\end{equation*}
For instance, considering $X^\star~=~(3, 3, 4, 1, 2, 1, 2, 2)$, $L~=~(2, 2, 3, 1, 2, 2, 3, 1)$ and $ C~=~(1, 2, 2, 1, 3, 2, 3, 4) $ for the given network in Fig.~\ref{fig3:benchmark}, for $ P_1 = \{a_1, a_4, a_7\}$, one has $ KP_1(X^\star) = 1 $, $ LP_1 = 6 $ and $ CP_1 = 5 $.
\begin{lemma}
	If $ P $ and $ P' $ are two paths in a network such that $ P\subset P' $, then:
	
	1) The lead time of $ P' $ is greater than $ P $,
	
	2) the transmission cost of $ P'  $ is greater than $ P $, and
	
	3) the capacity of $ P' $ under any SSV of $ X $ is less than $ P $.
\end{lemma}
\begin{proof}
As $ P' $ contain all the arcs of $ P $ and may have some other arc that does not belong to $ P $, the result directly is deduced from the definitions.
\end{proof}
\begin{proposition}
	For any path $ P $ in a network there is an MP, say $ P_1 $, such that $ P_1\subset P $.
\end{proposition}
\begin{proof}
	If $ P $ is an MP, then there is nothing left to prove. If not, then it has a proper subset $ P_1 $, which is still a path. If $ P_1 $ is an MP, the proof is complete. If not, one may repeat this argument. As the number of arcs in any path is limited, this repetition stops at some MP, completing the proof.
\end{proof}
The lemma and proposition above show that the paths with minimum transmission cost, maximum capacity, or best lead time are indeed among the MPs of the network. Therefore, we turn our attention to the MPs in the rest of this work.

\subsection{Basic results}
Here, some results are given based on which an algorithm is proposed in the following section to address the problem.

Let $ P_1, P_2, \cdots, P_q $ are all the minimal paths in the network, and hence $ q $ is the number of MPs. Following a similar argument, which was used for an arc in the preceding section, one can conclude that the transmission time to send $ d $ units of data through MP, $P_j= \{a_{j_1}, a_{j_2}, \cdots, a_{j_{nj}}\}$ under the SSV of $X$ is equal to
\begin{equation}\label{path-lead-time-formula}
	\xi(d, X, P_j)=LP_j+\lceil \frac{d}{KP_j(X)}\rceil.
\end{equation}
Similarly, the transmission cost of sending $ d $ units of data through $ P_j $ is equal to
\begin{equation}
	\beta(d, P_j) = d\times CP_j = d\sum_{r=1}^{nj} c_{j_r}.
\end{equation}
For example, considering $X^\star = (3, 3, 4, 1, 2, 1, 2, 2)$, $L=(2, 2, 3, 1, 2, 2, 3, 1)$ and $ C = (1, 2, 2, 1, 3, 2, 3, 4) $,  the transmission time and cost to send $ d= 10 $ units of data through $ P_1 = \{a_1, a_4, a_7\} $ in Fig.~\ref{fig3:benchmark} under SSV of $X^\star$ are $ \xi(10, X, P_1) = LP_1+ \lceil \frac{10}{KP_1(X^\star)}\rceil = 6+\lceil 10/1\rceil = 16$ and $ \beta(d, P_1) =10CP_1= 50$.
It is obvious that if $ \beta(d, P_j)>b $, for $ j=1, 2, \cdots, q $, one cannot transmit $ d $ units of data through a single MP from the source to the destination. Therefore, in this work, we assume that $ \min\{\beta(d, P_j)\ |\ j=1, \cdots, q\}\leq b $.

\begin{lemma}\label{lem: trans-time-MP}
	Assume that $ X $ and $ Y $ are two SSVs such that $ X\leq Y $. Then, $ \xi(d, X, P_j)\geq \xi(d,Y, P_j) $, for any demand level of $ d $ and any $ j=1, 2, \cdots, q $.	
\end{lemma}
\begin{proof}
	Let $ d $ be any given demand level and $ P_j $ any MP in the network. As $ X\leq Y $, one has $ KP_j(X)\leq KP_j(Y) $, and  according to Eq.~\eqref{path-lead-time-formula} it is concluded that $ \xi(d, X, P_j)\geq \xi(d,Y, P_j) $. 
\end{proof}
On the one hand, for sending $ d $ units of data from the source to the destination, the transmission time depends on the current SSV while the transmission cost does not. 
If for an SSV of $ X $ and an MP, say $ P_j $, we have $ \xi(d, X, P_j)\leq T $ but $ d\times CP_j>b $, then it is vivid that one cannot transmit $ d $ units of data through $ P_j $ within budget of $ b $. On the other hand, the data should be sent through a single MP, and that it is always preferable to have the minimum possible transmission time. Therefore, one can integrate the transmission time and cost and define 
the following function on SSV $ X $ related to transmitting $ d $ units of data from the source node to the destination node through a single MP.
\begin{equation}\label{eq: transmission-time}\small
	\Xi(d, X) = \min\{\xi(d,X, P_j)\ |\ \beta(d, P_j)\leq b,\ \text{ for }  j = 1, 2, \cdots, q\}
\end{equation}
Now, one observes that for any SSV of $ X $ with $ \Xi(d, X)\leq T $, there is at least one MP through which $ d $ units of data can be transmitted within $ T $ units of time and budget of $ b $.
\begin{theorem}\label{theorem: trans-time}
	If $ X $ and $ Y $ are two SSVs such that $ X\leq Y $, then, $ \Xi(d, X)\geq \Xi(d,Y) $.
\end{theorem}
\begin{proof}
	It is a direct result of Lemma~\ref{lem: trans-time-MP} and Eq.~\eqref{eq: transmission-time}.
\end{proof}

We recall that the main aim in this work is to assess the performance of a network, that is, to compute the exact network reliability, which is denoted by \rdts and is the probability of transmitting $ d $ units of data from a source node to a destination node through a single MP within the time of $ T $ and budget of $ b $. 
Letting $ \Omega_d = \{X\leq M\ |\ \Xi(d,X)\leq T\} $, one observes that \rdt $= \Pr\{X\ |\ X\in \Omega_d\} $. Now, let $ \Omega_{d, \min} = \{X^1, X^2, \cdots, X^{\sigma}\} $ be the set of all the minimal vectors in $ \Omega_d  $ and  $E_r=\{X|X\geq X^r\}$, for $r=1,2,\cdots,\sigma$. From Theorem~\ref{theorem: trans-time}, one sees that $\Omega_d =\cup_{1\leq r\leq \lambda}E_r$. Therefore, the system reliability can be computed by using some exact methods such as the inclusion-exclusion principle~\cite{forghani2014ieee} and the sum of disjoint products~\cite{balan2003preprocessing},~\cite{forghani2021chapter}. For example, using the inclusion-exclusion method, we have

\begin{multline}\label{eq: reliability}
\hspace{-3mm}	R_{(d,T, b)}=\Pr(\cup_{r=1}^\sigma E_r)=\sum^\sigma_{r=1}\Pr(E_r)-\sum^\sigma_{j=2}\sum^{j-1}_{r=1}\Pr(E_r\cap E_j)\\
	+...+(-1)^{\sigma+1}\Pr(\cap^\sigma_{r=1}E_r), 
\end{multline}

where 
\begin{align*}
	\Pr(E_r)=\sum_{X\in E_r} \Pr(X),\ \textnormal{ and }\Pr(X)= \prod^m_{i=1}\Pr(x_i).
\end{align*}

Now that we know how to compute the network reliability from the set of $ \Omega_d $, the rest is to determine such a set discussed in the next section.

\section{The proposed algorithm}\label{sec-algorithm}

We start the discussion here with two definitions.
\begin{definition}\label{dtb-can}
	An SSV of $ X $ is called a \dtbs candidate if $ \Xi(d, X)\leq T $ (see Eq.~\eqref{eq: transmission-time}).
\end{definition}

\begin{definition}\label{dtb-real}
	An SSV of $ X $ is called a (real) \dtbs  if $ \Xi(d, X)\leq T $ and $ \Xi(d, Y)>T $ for any SSV of $ Y<X $.
\end{definition}

\begin{lemma}\label{lem: minimal-candidates}
	Every (real) \dtbs is a candidate, and if $ \Omega_d $ is the set of all the \dtbs candidate, then  $ \Omega_{d, \min} $ is the set of all the \dtb s.
\end{lemma}
\begin{proof}
	The first part is clear. Let $X\in \Omega_{d, \min} $. Assume that $ Y<X $ is an SSV. If $ \Xi(d, Y)\leq T $, then $ Y\in \Omega_d $ which contradicts with $ X $ being in $ \Omega_{d, \min} $. Hence, $ \Xi(d, Y)>T $, and so $ X $ is a (real) \dtb. Now, let $ X $ be a \dtb. Thus, there is no $ Y<X $ with $ \Xi(d, Y)\leq T $, and so $ X\in \Omega_{d, \min} $.
\end{proof}

This lemma shows that one needs to find the minimal vectors among all the \dtbs candidates. An elementary approach is to find all the candidates and then determine the minimal vectors by comparing all the candidates. However, both finding the candidates and comparing them are very expensive and time-consuming operations. Assume that the data is transmitting through $ P_j $. In such a case, one can set the capacity of every arc which does not belong to $ P_j $ to zero.
Since the data is transmitting through only the arcs in $ P_j $, it does not affect the transmission cost or time. However, it decreases the current SSV. As we seek the minimal candidates, we need to find out the minimum possible capacity for each MP by which $ d $ units of data can be transmitted through the MP within $ T $ units of time and budget of $ b $. As the transmission cost is independent of the SSVs, one can calculate the transmission cost of all the MPs and remove every MP whose transmission cost exceeds $ b $. This way, there is no need to check the budget constraint anymore.
Moreover, if the lead time of an MP is greater than or equal to $ T $, then its transmission time for sending any amount of data will exceeds $ T $. Therefore, one can remove the MPs whose lead times are not less than $ T $. After removing these MPs, the following lemma provides the minimum possible capacity for each MP.

\begin{lemma}\label{lem: smallest-capacity}
	Let $ P_j $ be an MP such that $ CP_j\leq \frac bd $ and $ LP_j<T $. Then, $ \alpha_j = \lceil \frac{d}{T -LP_j}\rceil $ is the minimum possible capacity for $ P_j $ such that $ d $ units of data can be transmitted through it within $ T $ units of time and budget of $ b $. 
\end{lemma}
\begin{proof}
	Assume that $ X $ is an SSV such that $ KP_j(X)= \alpha_j = \lceil \frac{d}{T -LP_j}\rceil$. Then, we have 
	\begin{equation*}
		\xi(d, X, P_j) = LP_j+\lceil \frac{d}{\alpha_j}\rceil\leq LP_j+T-LP_j=T.
	\end{equation*}
	Hence, as $ CP_j\leq \frac bd $, one can send $ d $ units of data through this MP within time $ T $ and budget $ b $. Now, assume that $ Y $ is another SSV such that $ KP_j(Y)=\alpha>\alpha_j $. In this case, we have 
	\begin{equation*}
		\xi(d, Y, P_j) = LP_j+\lceil \frac{d}{\alpha}\rceil> LP_j+T-LP_j=T.
	\end{equation*}
	Thus, it is not possible to do the same under $ Y $ which completes the proof.
\end{proof}
As a results of this lemma and according to the definition of the capacity of an MP, one can set a vector corresponding to every MP as follows.\\

\subsubsection*{Setting the minimum SSV corresponding to an MP:} Suppose that $ P_j $ is an MP such that $ CP_j\leq \frac bd $ and $ LP_j<T $ and let $ \alpha_j = \lceil \frac{d}{T -LP_j}\rceil $. If $ \alpha_j\leq KP_j(M) $, then set the vector $ X_{P_j}=(x_1, x_2, \cdots, x_m) $ as follows.
\begin{align}
	\textnormal{For } i = 1, 2, \cdots, m, \ x_i = 
	\begin{cases}
		\alpha_j & \textnormal{if }\ a_i\in P_j,\cr
		0 & \textnormal{if }\ a_i\notin P_j.
	\end{cases}
\end{align}
The first condition to check is whether this vector is less than or equal to $ M $. In other words, if $ \alpha_j>KP_j(M)$, then it is impossible to have such a vector as an SSV for the network. It is why the condition of $ \alpha_j\leq KP_j(M) $ is necessary for the above approach. 
According to Lemma~\ref{lem: smallest-capacity}, one observes that $ \xi(d, X_{P_j}, P_j)\leq T $ and as $ CP_j\leq \frac bd $, it is deduced that $ \Xi(d, X_{P_j})\leq T $, and hence $ X_{P_j} $ is a \dtbs candidate. 
As the capacity of any arc not belonging to $ P_j $ is zero in $ X_{P_j} $, one sees that the capacity of $ P_j $ under any SSV, $ Y<X $, is less than $ \alpha_j $, and accordingly $ \Xi(d, Y)>T $. It shows that $ X_{P_j} $ is a minimal candidate, and thus it is a (real) \dtbs according to Lemma~\ref{lem: minimal-candidates}. As a result, one can determine the corresponding \dtb s to all the MPs which satisfy the budget constraint directly without generating any extra candidates. Now, we are at the point to state the proposed algorithm.\\

\noindent\textbf{Algorithm~1} (assessing the reliability of smart grid communication network under time and budget constraints)

\noindent\textbf{Input:} A smart grid communication network $ G(N, A, M, L, C) $ with a demand level of $ d $, time limit of $ T $, budget limit of $ b $, and all the MPs, say $ P_1, P_2, \cdots, P_q $.

\noindent\textbf{Output:} The system reliability, that is, \rdt.\\

\noindent\textbf{Step~1.} Calculate $ CP_j = \sum_{a_r\in P_j} c_{r} $ and $ LP_j =\sum_{a_r\in P_j} l_{r}$, for $ j=1, 2, \cdots, q $. Let $ J = [j\ |\ CP_j\leq \frac bd\ \&\ LP_j< T] $ and $ |J|=k $, that is, $ k $ is the number of MPs which satisfy the conditions.


\noindent\textbf{Step~2.} For each $ r=1, 2, \cdots, k $, let $ j = J(r) $, that is, the $ r $th index in $ J $, and $ \alpha_j = \lceil \frac{d}{T -LP_j}\rceil $. If $ \alpha_j\leq KP_j(M) $, then set $ X_{P_j}=(x_1, x_2, \cdots, x_m) $ as follows.
\begin{align*}
	\textnormal{For } i = 1, 2, \cdots, m, \ x_i = 
	\begin{cases}
		\alpha_j & \textnormal{if }\ a_i\in P_j,\cr
		0 & \textnormal{if }\ a_i\notin P_j.
	\end{cases}
\end{align*}

\noindent\textbf{Step~3.} Assume that $ \sigma\leq k $ is the number of obtained \dtbs in Step~2. If $ \sigma =0 $, then \rdt$ =0 $, otherwise let $E_j=\{X|X\geq X_{P_j}\}$, for $j=1,2,\cdots,\sigma$, and compute the system reliability by using the inclusion-exclusion technique through Eq.~\eqref{eq: reliability}.\\

According to the discussions in this section and the preceding ones, the following theorem is at hand.
\begin{theorem}
	Algorithm~1 evaluates the reliability of the given smart grid communication network correctly without generating any extra candidates.
\end{theorem}

\section{An illustrative example and complexity results}\label{sec-example-complexity}
\subsubsection{ A descriptive example} Assume that the corresponding communication network in a smart grid is the one given in Fig.~\ref{fig3:benchmark} with the provided arc data in Table~\ref{tab1-forghani} which includes the lead time, transmission cost, and probability distribution for the arcs' capacities. The administrator needs to know the reliability of this communication network for transmitting $ d=10 $ units of data from node $ 1 $ to node $ 5 $ within $ T=8 $ units of time and budget of $ b = 50 $. We use Algorithm~1 for this aim. From Table~\ref{tab1-forghani}, we have  $M=($5, 3, 4, 3, 2, 4, 5, 3$)$, $L=($2, 2, 3, 1, 2, 2, 3, 1$)$ and $ C = ($1, 2, 2, 1, 3, 2, 3, 4$) $.

\begin{table*}[t]
	\centering
	\caption{The arc data for Fig.~\ref{fig3:benchmark}.}\label{tab1-forghani}
	\begin{small}
		\begin{tabular}{lcccclcccc}
			\toprule
			Arc&Lead&Cost&Capacity&Probability&Arc&Lead&Cost&Capacity&Probability\\
			&time&&&&&time&&&\\
			\midrule			
			&&&$ 5 $&$ 0.7 $&&&&$ 2 $&$ 0.85 $\\
			&&&$ 4 $&$ 0.1 $&$ a_5 $&$ 2 $&$ 3 $&$ 1 $&$ 0.1 $\\
			$a_1$&$ 2 $&$ 1 $&$ 3 $&$ 0.05 $&&&&$ 0 $&$ 0.05 $\\
			&&&$ 2 $&$ 0.05 $&&&&&\\
			&&&$ 1 $&$ 0.05 $&&&&$ 4 $&$ 0.7 $\\
			&&&$ 0 $&$ 0.05 $&&&&$ 3 $&$ 0.1 $\\	
			&&&&&$ a_6 $&$ 2 $&$ 2 $&$ 2 $&$ 0.1 $\\	
			&&&$ 3 $&$ 0.8 $&&&&$ 1 $&$ 0.05 $\\
			$ a_2 $&$ 2 $&$ 2 $&$ 2 $&$ 0.1 $&&&&$ 0 $&$ 0.05 $\\
			&&&$ 1 $&$ 0.05 $&&&&&\\
			&&&$ 0 $&$ 0.05 $&&&&$ 5 $&$ 0.7 $\\
			&&&&&&&&$ 4 $&$ 0.1 $\\
			&&&$ 4 $&$ 0.7 $&$ a_7 $&$ 3 $&$ 3 $&$ 3 $&$ 0.05 $\\
			&&&$ 3 $&$ 0.1 $&&&&$ 2 $&$ 0.05 $\\
			$ a_3 $&$ 3 $&$ 2 $&$ 2 $&$ 0.1 $&&&&$ 1 $&$ 0.05 $\\
			&&&$ 1 $&$ 0.05 $&&&&$ 0 $&$ 0.05 $\\
			&&&$ 0 $&$ 0.05 $&&&&&\\
			&&&&&&&&$ 3 $&$ 0.8 $\\
			&&&$ 3 $&$ 0.8 $&$ a_8 $&$ 1 $&$ 4 $&$ 2 $&$ 0.1 $\\
			$ a_4 $&$ 2 $&$ 1 $&$ 2 $&$ 0.1 $&&&&$ 1 $&$ 0.05 $\\
			&&&$ 1 $&$ 0.05 $&&&&$ 0 $&$ 0.05 $\\
			&&&$ 0 $&$ 0.05 $&&&&&\\			
			\bottomrule
		\end{tabular}
	\end{small}
\end{table*}

\noindent\textbf{Step~1.} We calculate $ CP = (CP_1, \cdots, CP_9) =($3, 5, 9, 5, 5, 9, 6, 8, 8$) $ and $ LP=(LP_1, \cdots, LP_9)=($4, 6, 6, 5, 5, 5, 4, 8, 8$) $. As $ \frac bd = 5 $ and $ T=8 $, we have $ J = [1, 2, 4, 5] $ and $ k = 4 $.

\noindent\textbf{Step~2.} We calculate $ \alpha_1 = 3 $, $ \alpha_2 = 5 $, $ \alpha_4 = 4 $, $ \alpha_5 = 4 $, $ KP_1 = 4 $, $ KP_2 = 3 $, $ KP_4 = 3 $ and $ KP_5 = 3 $. As $ \alpha_j>KP_j $, for $ j =$ 2, 4, and 5, then we set only $ X_{P_1} =($3, 0, 0, 0, 0, 3, 0, 0$)$.

\noindent\textbf{Step~3.} As $ \sigma=1 $, we have $ E_1 = \{X\ |\ X\geq ($3, 0, 0, 0, 0, 3, 0, 0$)\} $, and accordingly \rdt$ =0.68 $ which is not so good. For sure if the administrator increase the time or budget limits, the network reliability may increase.\\

This example shows how efficiently the algorithm removes several MPs to lessen the required work. One notes that finding all the \dtbs candidates and removing the non-minimal ones requires much more time than the consumed time in Algorithm~1. It is worthy to see that according to $M=($5, 3, 4, 3, 2, 4, 5, 3$)$, there are 172800 possible SSVs for this network, whereas Algorithm~1 directly found the only existing $ (10, 8, 50) $-$ MP $.

\subsubsection{ Complexity results}
We first recall that $ m $ and $ q $ are respectively the numbers of arcs and MPs.
The time complexity of Step~1 to calculate $ CP_j $ and $ LP_j $, for $ j=1, \cdots, q $, and remove the ones which satisfy the desired conditions is $ O(mq) $. As $ k $ is the number of MPs in Step~2, the time complexity of this step to compute $ \alpha_j $s and set the vectors, if any, is $ O(mk) $. Finally, as $ \sigma $ vectors arrive at Step~3, the time complexity of calculating the union probability in this step is of order of $ O(m\sigma^2) $~\cite{balan2003preprocessing}. As $ \sigma^2>q\ \&\ k $, for large enough networks, the following theorem is at hand.

\begin{theorem}
	The time complexity of Algorithm~1 to find all the \dtb s  is $ O(mq) $ and to assess the reliability of the network is $ O(m\sigma^2) $, where $ m$ and $ q $  are respectively the number of arcs and MPs in the network, and $ \sigma $ is the number of obtained \dtb s in the algorithm. 
\end{theorem}

\section{Experimental results}\label{sec-experimental results}
To show the practical efficiency of our proposed algorithm in comparison with the available algorithms in the literature, we compare our algorithm and the one proposed by Lin~\cite{lin2003extend}. It is noted that the algorithm of~\cite{lin2003extend} serves for the cases without the budget limit, and hence we add a new step into the algorithm to  check all the solutions for the budget constraint. For more convenience and readability, we state Lin's algorithm with a newly added step here as Algorithm~2.\\

\noindent\textbf{Algorithm~2} (based on Lin's algorithm~\cite{lin2003extend} for assessing the reliability of MFNs under time and budget constraints)

\noindent\textbf{Input:} An MFN $ G(N, A, M, L, C) $ with a demand level of $ d $, time limit of $ T $, budget limit of $ b $, and all the MPs, say $ P_1, P_2, \cdots, P_q $.

\noindent\textbf{Output:} The system reliability, that is, \rdt.\\

\noindent\textbf{Step~1.} For each $ j=1, 2, \cdots, q $:

\textbf{1.1.} Find the minimum integer-valued number $ v_j $ such that $ LP_j+\lceil \frac{d}{v_j} \rceil\leq T $.

\textbf{1.2.} If $ v_j\leq KP_j(M) $, then set $ X_{P_j}=(x_1, x_2, \cdots, x_m) $ as follows.
\begin{align*}
\textnormal{For } i = 1, 2, \cdots, m, \ x_i = 
\begin{cases}
\alpha_j & \textnormal{if }\ a_i\in P_j,\cr
0 & \textnormal{if }\ a_i\notin P_j.
\end{cases}
\end{align*}

\noindent\textbf{Step~2.} (Newly added step to check the solutions for the budget constraint) Remove every solution $ X_{P_j} $ that satisfies $ d\times CP_j>b $.

\noindent\textbf{Step~3.} Assuming $ \sigma $ as the number of obtained solutions, if $ \sigma>0 $, let  $E_j=\{X|X\geq X_{P_j}\}$, for $j=1,2,\cdots,\sigma$, and compute \rdt $= \Pr(\cup_{r=1}^\sigma E_r) $, else \rdt$ =0 $.\\

As Step~3 in both algorithms~1 and~2 are to compute a union probability by using the available techniques in the literature, we compare the practicality of the algorithms based on their first two steps. Therefore, we use the algorithms to obtain all the SSVs under which $ d $ units of data can be sent from a source to a sink within a time of $ T $ and a budget of $ b $. This way, we compare the algorithms on the taken CPU seconds to solve the test problems. 


To generate the numerical results, both algorithms are implemented in the MATLAB programming environment. We note that for implementing Algorithm~2, we use our presented results of Lemma~\ref{lem: smallest-capacity}, and accordingly the obtained results on this algorithm here are expected to be better than what could have been. 
We employ the Pan European topology, which is a rather large-size benchmark, depicted in Fig.~\ref{fig4:pan-european}, as well as one thousand large-size randomly generated test problems to compare the practical efficiency of the algorithms. All the numerical results were made on a computer with Intel(R) Core(TM) i5-2400S Duo CPU 3.1 GHz, with 8 GB of RAM.

\begin{figure}
	\centering
	\includegraphics[width=\linewidth]{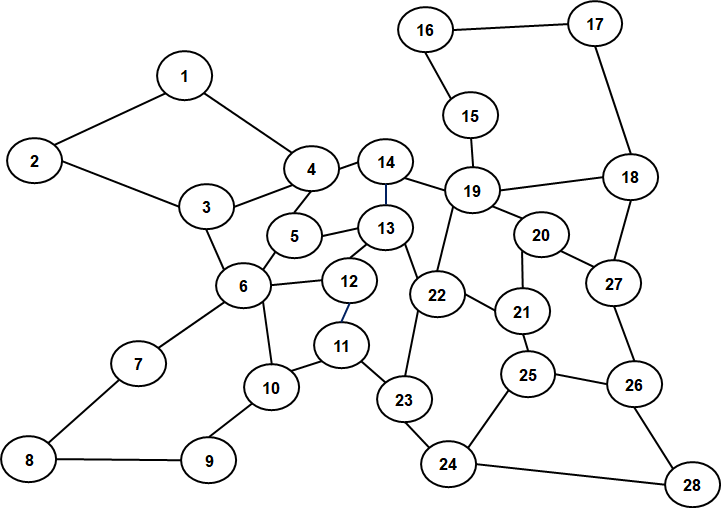}
	\caption{The Pan European network topology.} 
	\label{fig4:pan-european}
\end{figure}
For each network, the arcs' capacities, lead times, and transmission costs take random integer values  in the intervals $ [10, 50] $, $ [5, 10] $, and $ [5, 20] $,  respectively. 
We note that with large enough limits of time and budget, any SSV can be a solution and there is nothing more to be solved by the algorithms. Hence, to have some meaningful limits, we considered 
\begin{align*}
	T = \lceil \frac{LP_1+\cdots+LP_q}{q} \rceil \text{ and } b = \lceil \frac{d(CP_1+\cdots+CP_q)}{q} \rceil
\end{align*}

\noindent as the time and budget limits in each test problem, respectively.

For the Pan European benchmark, we first calculate
\begin{align*}
	d^\star = \lceil \frac{KP_1+\cdots+KP_q}{q} \rceil ,
\end{align*}

 \noindent and then consider all the ten different demand levels of $ d = d^\star - 5, d^\star-4, \cdots, d^\star, \cdots, d^\star+4$. This way, we have ten different cases related to the Pan European network. There are 28 nodes, 40 arcs and 1274 MPs in this network from node $ 1 $ to node $ 28 $ which make it a rather large-size network.  
 Table~\ref{tab:pan-european-forghani} provides the final results obtained on this network. The columns in this table are $ d $, the demand level, $ N_{h,A2} $, number of generated solutions in Algorithm~2 with transmission cost higher than~$ b $, $ N_{r,A1} $, number of removed MPs with unaffordable transmission time or cost in the outset of Algorithm~1, $ N_s $, number of solutions, $ t_1 $, the running time of Algorithm~1, $ t_2 $, the running time of Algorithm~2, and $ t_2/t_1 $ which is time ratio. The given CPU times in this table are in seconds. 
\begin{table}[t]
	\centering
	\caption{THE FINAL RESULTS ON PAN EUROPEAN BENCHMARK GIVEN IN FIG.~\ref{fig4:pan-european}.}\label{tab:pan-european-forghani}
	\begin{small}
		\begin{tabular}{ccccccc}
			\toprule
			$ d $&$ N_s $&$ N_{h, A2} $&$ N_{r,A1} $&$ t_1 $&$ t_2 $&$ t_2/t_1 $\\
			\midrule						
			6&560&714&78&0.0170&0.0208&1.2237\\			
			7&560&714&78&0.0129&0.0190&1.4796\\
			8&560&714&78&0.0127&0.0173&1.3624\\
			9&560&714&78&0.0122&0.0175&1.4306\\
			10&560&714&78&0.0128&0.0176&1.3765\\
			11&551&714&78&0.0122&0.0167&1.3721\\
			12&550&714&78&0.0130&0.0165&1.2756\\
			13&550&714&78&0.0126&0.0166&1.3150\\
			14&550&714&78&0.0122&0.0165&1.3563\\
			15&549&714&78&0.0124&0.0165&1.3246\\
			\midrule
			\multicolumn{4}{c}{Geo. Mean}&0.0130&0.0175&1.3471\\
			\bottomrule
		\end{tabular}
	\end{small}
\end{table}

Table~\ref{tab:pan-european-forghani} shows that Algorithm~1 removes 714 out of 1274 MPs in the outset for each case which is more than half of the MPs. Among these MPs, as seen from the fourth column, only 78 MPs are removed because of the high cost. Algorithm~1 solves all the cases (in average 1.34 times) faster than Algorithm~2 which shows the practical efficiency of our proposed algorithm. 

To have a more meaningful comparison of the algorithms, we use one thousand randomly generated test problems. To this end, we consider $ n = 11, 12, \cdots, 30 $ as the number of nodes and generate $ 50 $ random random networks for each case, making a total of $ 1000 $ test problems. To not have very complicated network, we consider two integer numbers $ f = 3(\lceil n/2 \rceil -1)$ and $ g = 25-\lceil n/2\rceil $, where $ n $ is the number of nodes in the network. Then, the number of arcs in each corresponding generated network take random integer values from the interval $ [f, f+g] $. We use the proposed algorithm in~\cite{forghani2014ieee} to generate the random networks. The arcs' data along with time and budget limits are determined randomly similar to the case of the Pan European network. Moreover, the demand level in each test problem is considered as geometry mean of the MPs' capacities, that is, 
\begin{align*}
	d = \lceil \frac{KP_1+\cdots+KP_q}{q} \rceil.
\end{align*}
Therefore, in some way, we have $ 20 $ cases with $ 50 $ test problems in each case. Table~\ref{tab:rand-networks-forghani} provides the average running times of each algorithm for these cases. In this table, $ n $, $ \tilde{t_1} $, and $ \tilde{t_2} $ denotes the number of nodes and the average running times of Algorithm~1 and 2 on the 50 corresponding networks with each case. The fourth column shows that Algorithm~1 solves  the problems of some cases even on average more than 160 times faster than Algorithm~2. Moreover, as the running times increase the proportional $ \tilde{t_2}/\tilde{t_1} $ increase as well which shows clearly that the practical efficiency of our proposed algorithm in comparison with Algorithm~2 enhances as the problem's size enlarges. 

\begin{table}[t]
	\centering
	\caption{THE RUNNING TIMES IN CPU SECONDS ON RANDOMLY GENERATED NETWORKS.}\label{tab:rand-networks-forghani}
	\begin{small}
		\begin{tabular}{cccc}
			\toprule
			$ n $&$ \tilde{t_1} $&$ \tilde{t_2} $&$ \tilde{t_2}/\tilde{t_1} $\\
			\midrule								
			11&0.0149&0.1121&7.5084\\
			12&0.0087&0.0412&4.7086\\
			13&0.0236&0.3042&12.867\\
			14&0.0222&0.2008&9.0581\\
			15&0.0563&1.4074&25.003\\
			16&0.0357&0.4426&12.412\\
			17&0.0676& 2.761&40.866\\
			18&0.0403&0.6165&15.303\\
			19&0.1634&10.557&64.624\\
			20&0.0406&0.8169&  20.1\\
			21&0.1498&8&53.41\\
			22&0.1027&2.7181&26.454\\
			23&0.3394&30.245&89.111\\
			24&0.1115&3.5526&31.849\\
			25& 0.225&14.608&64.936\\
			26&0.2003&9.7149&48.496\\
			27&0.3763&24.623&65.436\\
			28&0.1833&6.8589&37.421\\
			29&0.6918&111.06&160.54\\
			30&  0.29&14.799&51.031\\
			\bottomrule
		\end{tabular}
	\end{small}
\end{table}

In addition to Table~\ref{tab:rand-networks-forghani}, to see more intuitively, we consider the running times of the algorithms on these one thousand random test problems for producing the performance profile introduced by Dolan and Moré~\cite{dolan2002benchmarking}. In this performance profile, the ratio of the executing times of the algorithms versus the best ones are considered.  
Assuming $t_{i,1}$ and $t_{i,2}$, respectively, as the running times of algorithms~1 and 2, for $i = 1, 2,\cdots, 1000$, the performance ratios are $r_{i,j}~=~\frac{t_{i,j}}{\min\{t_{i,j}:\ j~=~1,~2\}}$, for $ j~=~1,~2$~\cite{dolan2002benchmarking}. 
For each algorithm, the performance  is determined by $Pr_j(\tau)~=~\frac{N_j}{1000}$, where $ N_j $ is the number of test problems for which $ r_{i,j}~\leq~\tau, \ i = 1, 2, \cdots, 1000 $. 
\begin{figure}
	\centering
	\includegraphics[width=\linewidth]{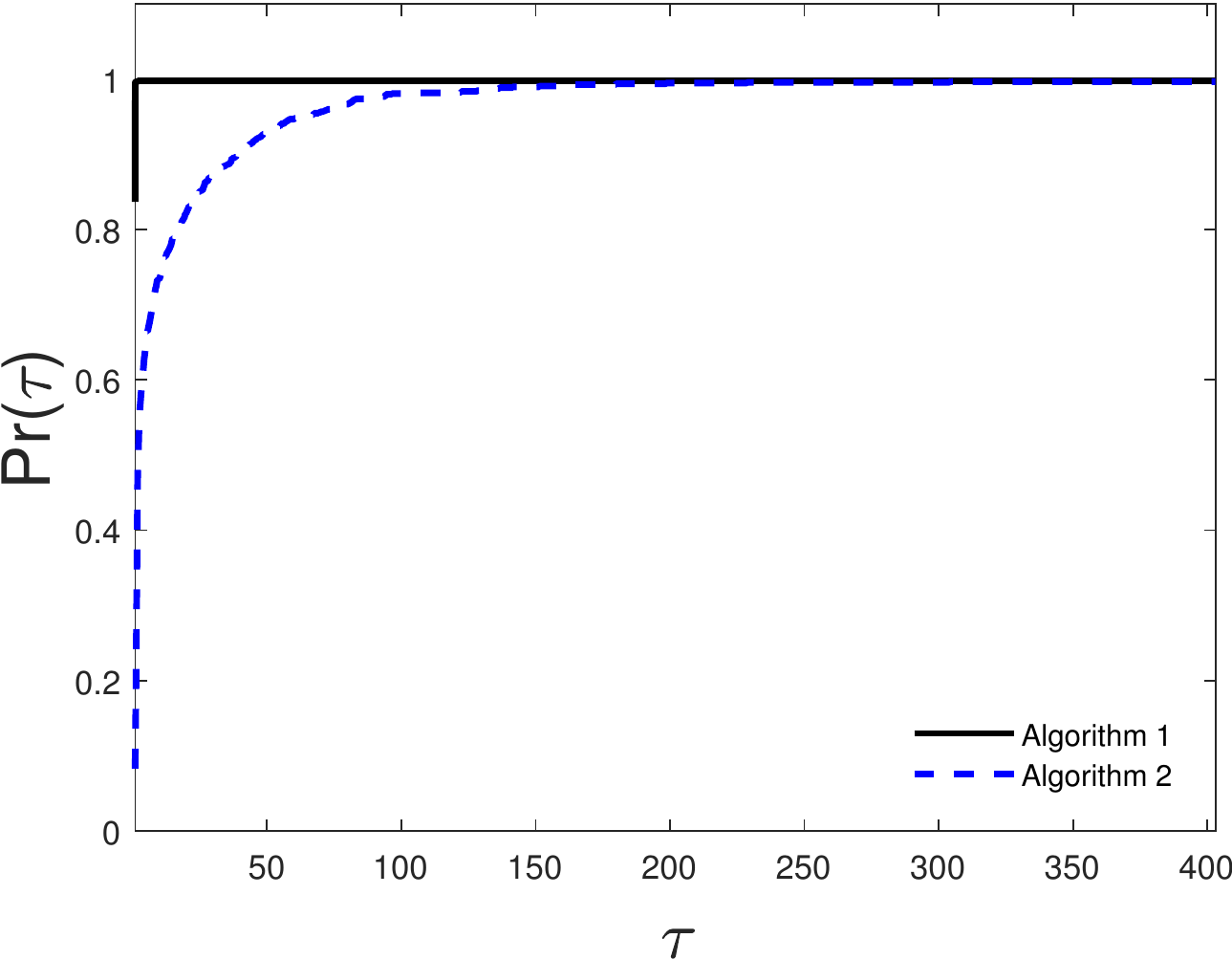}
	\caption{ CPU time performance profiles for Algorithms 1 and 2.} 
	\label{fig5:dolan-more}
\end{figure}
Fig.~\ref{fig5:dolan-more} depicts the result of this performance profile for the algorithms. The figure shows that Algorithm~1 solves all the test problems faster than Algorithm~2. Observing at $ \tau=50 $ in the figure, one can see that around $ 10\% $ of the test problems have been solved more than 50 times faster by Algorithm~1. Moreover, the horizontal axis shows that there exist some problems which have been solved by Algorithm~1 more than 400 times faster. In this profile, the algorithm whose performance diagram lies above the other is preferred~\cite{dolan2002benchmarking}.

All the numerical results including the tables~\ref{tab:pan-european-forghani} and~\ref{tab:rand-networks-forghani} and Fig.~\ref{fig5:dolan-more} show the superiority of Algorithm~1 to Algorithm~2, and that as the network's size grows this superiority increases.

\section{Conclusions}\label{sec-conclusions} 
Communication networks play a crucial role in smart grids, and hence assessing the performance of such networks is of great importance.  Moreover, the time and budget constraints are usually notable challenges in real-world systems. Therefore, taking into account both time and budget limits, to evaluate the performance of such networks, we considered a reliability index, which is the probability of transmitting $ d $ units of data from a source to a destination through a single path within $ T $ units of time and budget of $ b $. Some results were presented based on which an efficient algorithm was proposed to address the problem. The algorithm was illustrated by using a benchmark network example. We also provided the complexity results. Moreover, several experimental results were provided by employing the Pan European topology and one thousand randomly generated networks. The numerical results clearly showed that our proposed algorithm outperforms some available algorithm in the literature, and that the superiority of our algorithm enhances as the network's size grows.

\section*{Acknowledgments}
The second author thanks CNPq (grant 306940/2020-5) for supporting this work.

\ifCLASSOPTIONcaptionsoff
  \newpage
\fi



%
%
%

\bibliographystyle{ieeetr}
\bibliography{Article-Forghani-IEEE_Trans-Reliab}

%

\begin{IEEEbiography}[{\includegraphics[width=1in,height=1.25in,clip,keepaspectratio]{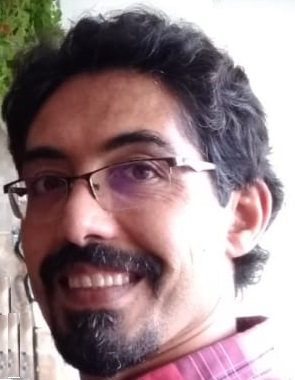}}]{Majid Forghani-elahabad}
is an assistant professor of Applied Mathematics in the Federal University of ABC. He received his B.Sc. and M.S. degrees in Applied Mathematics in 2005 and 2008, respectively, and his Ph.D. degree in Operations Research in 2014. He was a postdoctoral researcher at the University of Sao Paulo and the Federal University of ABC (UFABC) and a lecturer at Sharif University of Technology and several other universities in Iran. He has served as a reviewer of several journals including Reliability Engineering \& System Safety and Computers \& Industrial Engineering. He is currently a researcher of level 2 in Brazil and is working on system reliability, network optimization and optical networks.
\end{IEEEbiography}

%
%




\end{document}